\def\year{2017}\relax
\documentclass[letterpaper]{article}
\usepackage{aaai17}
\usepackage{times}
\usepackage{helvet}
\usepackage{courier}
\usepackage[utf8]{inputenc} % allow utf-8 input
\usepackage[T1]{fontenc}    % use 8-bit T1 fonts
\usepackage{url}            % simple URL typesetting
\usepackage{booktabs}       % professional-quality tables
\usepackage{amsmath}
\usepackage{amssymb}
\usepackage{amsthm}
\usepackage{amsfonts}       % blackboard math symbols
\usepackage{nicefrac}       % compact symbols for 1/2, etc.
\usepackage{microtype}      % microtypography
\usepackage{comment}
\usepackage{multicol,lipsum}
\usepackage{wrapfig}
\usepackage{graphicx}
\usepackage{balance}  % for  \balance command ON LAST PAGE  (only there!)
\usepackage{algorithmicx}              % list packages between braces
\usepackage{algpseudocode} %[noend]
\usepackage{algorithm}
\usepackage{color}
\usepackage{subfig}
\usepackage{enumerate}
\usepackage{placeins}
\usepackage{paralist}

\theoremstyle{plain}% default
\newtheorem{theorem}{Theorem}%[section]
\newtheorem*{theorem*}{Theorem}%[section]
\newtheorem{proposition}{Proposition}%[section]
\newtheorem{corollary}{Corollary}%[section]
\theoremstyle{definition}
\newtheorem{definition}{Definition}%[section]
\theoremstyle{remark}

\newcommand{\beq}{\begin{eqnarray}}
\newcommand{\eeq}{\end{eqnarray}}

\newcommand\nc\newcommand
\nc\bfa{{\boldsymbol a}}\nc\bfA{{\boldsymbol A}}\nc\cA{{\mathcal A}}
\nc\bfb{{\boldsymbol b}}\nc\bfB{{\boldsymbol B}}\nc\cB{{\mathcal B}}
\nc\bfc{{\boldsymbol c}}\nc\bfC{{\boldsymbol C}}\nc\cC{{\mathcal C}}
\nc\sC{{\mathscr C}}
\nc\bfd{{\boldsymbol d}}\nc\bfD{{\boldsymbol D}}\nc\cD{{\mathcal D}}
\nc\bfe{{\boldsymbol e}}\nc\bfE{{\boldsymbol E}}\nc\cE{{\mathcal E}}
\nc\bff{{\boldsymbol f}}\nc\bfF{{\boldsymbol F}}\nc\cF{{\mathcal F}}
\nc\bfg{{\boldsymbol g}}\nc\bfG{{\boldsymbol G}}\nc\cG{{\mathcal G}}
\nc\bfh{{\boldsymbol h}}\nc\bfH{{\boldsymbol H}}\nc\cH{{\mathcal H}}
\nc\bfi{{\boldsymbol i}}\nc\bfI{{\boldsymbol I}}\nc\cI{{\mathcal I}}
\nc\bfj{{\boldsymbol j}}\nc\bfJ{{\boldsymbol J}}\nc\cJ{{\mathcal J}}
\nc\bfk{{\boldsymbol k}}\nc\bfK{{\boldsymbol K}}\nc\cK{{\mathcal K}}
\nc\bfl{{\boldsymbol l}}\nc\bfL{{\boldsymbol L}}\nc\cL{{\mathcal L}}
\nc\bfm{{\boldsymbol m}}\nc\bfM{{\boldsymbol M}}\nc\cM{{\mathcal M}}
\nc\bfn{{\boldsymbol n}}\nc\bfN{{\boldsymbol N}}\nc\cN{{\mathcal N}}
\nc\bfo{{\boldsymbol o}}\nc\bfO{{\boldsymbol O}}\nc\cO{{\mathcal O}}
\nc\bfp{{\boldsymbol p}}\nc\bfP{{\boldsymbol P}}\nc\cP{{\mathcal P}}
\nc\bfq{{\boldsymbol q}}\nc\bfQ{{\boldsymbol Q}}\nc\cQ{{\mathcal Q}}
\nc\bfr{{\boldsymbol r}}\nc\bfR{{\boldsymbol R}}\nc\cR{{\mathcal R}}
\nc\bfs{{\boldsymbol s}}\nc\bfS{{\boldsymbol S}}\nc\cS{{\mathcal S}}
\nc\bft{{\boldsymbol t}}\nc\bfT{{\boldsymbol T}}\nc\cT{{\mathcal T}}
\nc\bfu{{\boldsymbol u}}\nc\bfU{{\boldsymbol U}}\nc\cU{{\mathcal U}}
\nc\bfv{{\boldsymbol v}}\nc\bfV{{\boldsymbol V}}\nc\cV{{\mathcal V}}
\nc\bfw{{\boldsymbol w}}\nc\bfW{{\boldsymbol W}}\nc\cW{{\mathcal W}}
\nc\bfx{{\boldsymbol x}}\nc\bfX{{\boldsymbol X}}\nc\cX{{\mathcal X}}
\nc\bfy{{\boldsymbol y}}\nc\bfY{{\boldsymbol Y}}\nc\cY{{\mathcal Y}}
\nc\bfz{{\boldsymbol z}}\nc\bfZ{{\boldsymbol Z}}\nc\cZ{{\mathcal Z}}

\nc{\remove}[1]{}

\nc\diff{{\mathrm d}}
\nc\e{{\mathrm e}}
\nc\calC{{\mathcal C}}
\newcommand{\h}{h_\mathrm{B}}

\newcommand{\expect}{{\mathbb E}}

\def\h_q{\qopname\relax{no}{h_q}}

\newcounter{ALC@tempcntr}% Temporary counter for storage

\newcommand{\cc}{{\sf Crowd-ER}}

\DeclareMathAlphabet{\mathpzc}{OT1}{pzc}{m}{it}

\begin{document}
\title{A Theoretical Analysis of First Heuristics of Crowdsourced Entity Resolution}
%\author{
%  Arya Mazumdar \\
%  College of Information and Computer Science\\
%  University of Massachusetts at Amherst\\
%  Amherst, MA 01003 \\
%  \texttt{arya@cs.umass.edu} \\
% \And
% Barna Saha \\
%  College of Information and Computer Science\\
%  University of Massachusetts at Amherst\\
%  Amherst, MA 01003 \\
%  \texttt{barna@cs.umass.edu} 
% }
\author{
Arya Mazumdar \and Barna Saha \\
College of Information \& Computer Sciences\\
University of Massachusetts Amherst \\ 
\texttt{\{arya,barna\}@cs.umass.edu}
}
\maketitle

\begin{abstract}
Entity resolution (ER) is the task of identifying all records in a database that refer to the same underlying entity, and are therefore duplicates of each other. Due to inherent ambiguity of data representation and poor data quality, ER is a challenging task for any automated process. As a remedy, human-powered ER via crowdsourcing has become popular in recent years. Using crowd to answer queries is costly and time consuming. Furthermore, crowd-answers can often be faulty. Therefore, crowd-based ER methods aim to minimize human participation without sacrificing the quality and use a computer generated similarity matrix actively. While, some of these methods perform well in practice, no theoretical analysis exists for them, and further their worst case performances do not reflect the experimental findings. This creates a disparity in the understanding of the popular heuristics for this problem. In this paper, we make the first attempt to close this gap. We provide a thorough analysis of the prominent heuristic algorithms for crowd-based ER. We justify experimental observations with our analysis and information theoretic lower bounds.\end{abstract}

\section{Introduction} %\textcolor{red}{Can we cite Larsen and Rubin somehow?}
Entity resolution (ER, record linkage, deduplication, etc.) seeks
to identify which records in a data set refer to the same underlying
real-world entity \cite{fellegi1969theory,elmagarmid2007duplicate,getoor2012entity,larsen2001iterative,christen2012data}. Our ability to represent %and misrepresent
information about real-world entities in very diverse ways makes this a complicated problem.
For example, collecting profiles of people and businesses, or specifications
of products and services from websites and social media
sites can result in billions of records that need to be resolved. These
entities are identified in a wide variety of ways, complicated further by language ambiguity, poor data entry, missing values, changing attributes and formatting issues.  ER is a fundamental task in data processing with wide-array of applications. There is a huge literature on ER techniques; many include machine learning algorithms, such as decision trees, SVMs, ensembles of classifiers, conditional random fields, unsupervised learning etc. (see \cite{getoor2012entity} for a recent survey). Yet, ER remains a demanding task for any automated strategy yielding low accuracy.

ER can be cast as a  clustering problem. Consider a set of $n$ elements $V$ that must be clustered into  $k$ disjoint parts $V_i, i=1,2,\ldots,k$.  The  true underlying clusters  $V_i \in [n]$, $i \in [1,k]$ are unknown to us, and so is $k$.  Each of these $V_i$s represents an entity. % that we want to discover.  
 Each element $v \in V$ has a set of attributes. A similarity function is used to estimate the similarity of the attribute sets of two nodes $u$ and $v$. If $u$ and $v$ represent the same entity, then an ideal similarity function will return $1$, and if they are different, then it will return $0$. 
  However, in practice, it is impossible to find an ideal similarity function, or even a function close to it. Often, some attribute values may be missing or incorrect, and that leads to similarity values that are noisy representation of the ideal similarity function. 
% Any automated process that uses $W$ to identify the latent clusters $V_1,V_2,\dots,V_k$ may make mistakes if $f_r$ and $f_g$ are close (in terms of some metric). 
Any automated process that uses such similarity function is thus prone to make errors.
To overcome this difficulty, a relatively recent line of works  propose to use human knowledge via crowdsourcing to boost accuracy of ER \cite{dkmr:14,fss:16,DBLP:conf/icde/VerroiosG15,DBLP:journals/corr/GruenheidNKGK15,wang2012crowder,wang2013leveraging,vesdapunt2014crowdsourcing,DBLP:conf/nips/YiJJJY12,whang2013question}. Human based on domain knowledge can match and distinguish entities with complex representations, where automated strategies fail.
 
 \begin{figure*}[ht]
\centering
\includegraphics[width=0.7\textwidth]{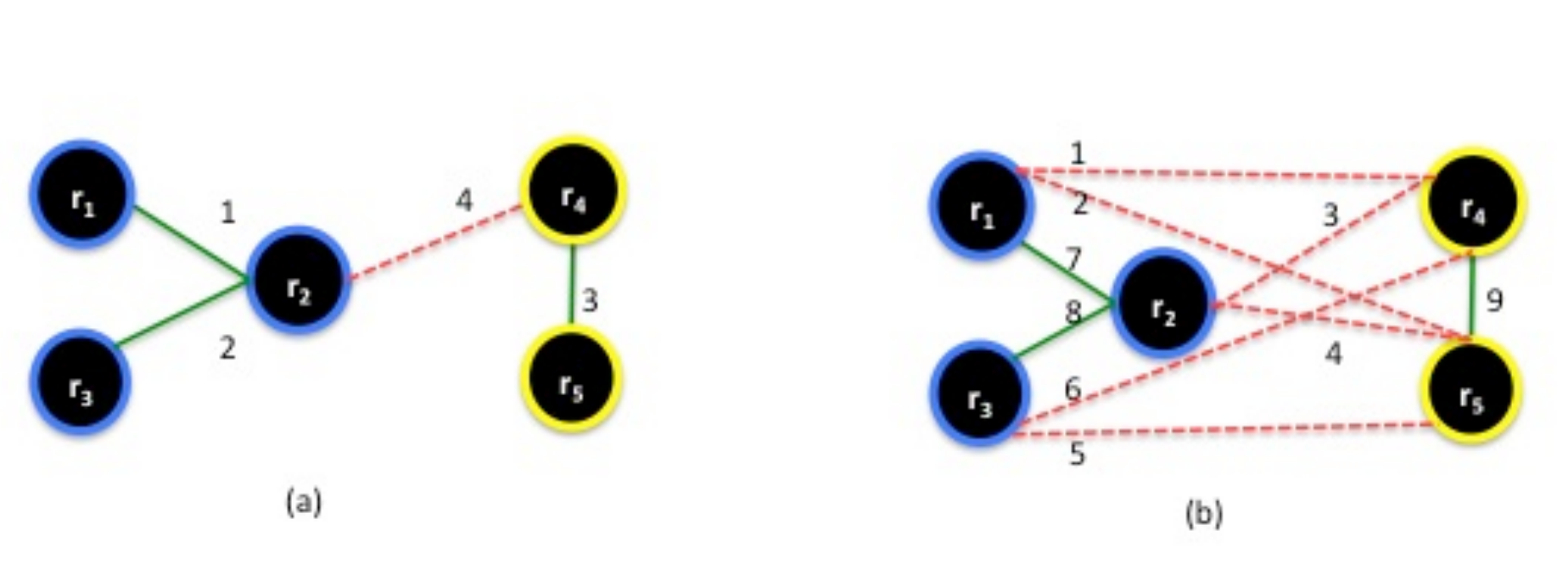}
\vspace{-0.2in}
\caption{Record pairs connected by green (resp.\ red) edges are ``matching'' (resp.\ ``non-matching'') in the real world. The numbers on the edges indicate the ordering of queries. While querying strategy of (a) only results in $4$ queries, the querying strategy of (b) results in $9$ queries. }
\label{fig:example}
\end{figure*}

\noindent{\bf Motivating example.} Consider the following illustrative example shown in 
Figure~\ref{fig:example}.
The Walt Disney, commonly known as Disney, is an American multinational media and entertainment company that owns and licenses $14$ theme parks around the world.
\footnote{
\url{https://en.wikipedia.org/wiki/The_Walt_Disney_Company}}
Given the six places 
($r_1$)~{\sf Disney World},
($r_2$)~{\sf Walt Disney World Resort}, 
($r_3$)~{\sf Walt Disney Theme Park, Orlando},
($r_4$)~{\sf Disneyland},
($r_5$)~{\sf Disneyland Park},
humans can determine using domain knowledge that these correspond to 
two entities: $r_1, r_2, r_3$ refer to one entity, and $r_4, r_5$ refer to a second entity.

Answering queries by crowd could be time-consuming and costly. Therefore, a crowd based ER strategy must attempt to minimize the number of queries to the oracle while  resolving the clusters exactly. Having access to ideal crowd answers, a good ordering of comparing record pairs is $(r_1,r_2)$, $(r_2,r_3)$, $(r_4,r_5)$, $(r_1,r_5)$. After the first three pairs have been compared, we can safely infer as ``matching'' the remaining pair $(r_1,r_3)$ leveraging transitive relations. After the last pair in the ordering has been compared, we can safely infer as ``non-matching'' all the remaining pairs $(r_1,r_4)$, $(r_2,r_4)$, $(r_2,r_5)$, $(r_3,r_4)$, $(r_3,r_5)$ in the database.

 The work by Wang et al. \cite{wang2013leveraging} was among the first few \cite{wang2012crowder,demartini2012zencrowd,whang2013question} to propose the notion of hybrid human-machine approach for entity resolution. Moreover, it is the first paper to leverage the transitive relationship among the entities to minimize the number of queries which has since become a staple in every follow-up work on this topic \cite{fss:16,DBLP:conf/icde/VerroiosG15,DBLP:journals/corr/GruenheidNKGK15,vesdapunt2014crowdsourcing}. Assuming there is an oracle, an abstraction of a crowd-sourcing platform that can correctly answer questions of the form {\it ``Do records $u$ and $v$ refer to the same entity?''}, they presented a new algorithm for crowd-sourced ER.  To minimize the number of queries to the crowd oracle, Wang et al. utilizes the transitive relation in
which known match and non-match labels on some record pairs can
be used to automatically infer match or non-match labels on other
record pairs. 
In short, the heuristic algorithm by Wang et al. does the following: it orders the tuples (record pairs/edges)  in nonincreasing order of similarity, and query any edge according to that order whenever the right value of that edge cannot be transitively deduced from the already queried/inferred edges so far.

While the crowd-sourcing algorithm of Wang et al. works reasonably well on real datasets,  theoretical guarantees for it was not provided. However, in \cite{vesdapunt2014crowdsourcing}, Vesdapunt et al. showed that in some instances this algorithm can only give an $\Theta(n)$ approximation, that is when an optimum algorithm may require $c$ queries,  Wang et al.'s algorithm can require $\Theta(cn)$ queries. 

Vesdapunt et al. 
proposed an algorithm  that proceeds in the following iterative manner. In each round, an element to be clustered is compared with
one representative of all the existing clusters. The order of these comparisons is defined by a descending order of the similarity measures. As soon as a positive query result is found the element is assigned to the corresponding cluster and the algorithm moves to the next round with a new element. It is  easy to see that in the worst case the number of queries made by the algorithm is $nk$, where $n$ is the number of 
elements and $k$ is the number of clusters. It also follows that this is at least an $O(k)$ approximation. 

Note that \cite{wang2013leveraging,vesdapunt2014crowdsourcing} consider the answers of queries are correct as an ideal crowd abstraction - this can often be guaranteed via majority voting. But it is unclear that how the quality of the similarity measurements affects the total number of queries. Indeed, in typical datasets, the performances of the algorithms of 
Wang et al. and Vesdapunt et al. are quite similar, and they are much better than their worst case guarantees that do not take into account the existence of any meaningful similarity measures. This means the presence of
the similarity measures  helps reduce the query complexity significantly. Is there a way to theoretically establish that and come up with guarantees that match the experimental observations?

It is of paramount interest to characterize the query complexity (number of questions asked to the crowd) of these popular heuristics
and  come up with  algorithms that minimize such complexity. The query complexity is directly proportional to 
the overall cost of a crowd-based algorithm, due to the fact that crowd questions are time-consuming and in many times involve  compensations. Designing a strategy that would minimize the query complexity can directly be seen as alternatives to {\em active learning} problem with minimum labeling requirements \cite{Sarawagi:2002,Bellare:2012}. From the perspective of lower bounding the query complexity, ER  can be seen as a  {\em reinforcement learning} problem. Indeed, in each step of assigning a record to one of the underlying entities, a query must be made wisely so that under any adversarial configurations, the total number of queries remain small. %This is similar philosophically to
% {\em  regret minimization} in multi-arm bandit problems \cite{auer2002nonstochastic,cesa2006prediction} as reflected in our lower bounding technique. 
%\vspace{0.1in}
 
\noindent{\bf Contributions.}
In this paper we assume the following model for the similarity measurements.  Let $W = \{w_{u,v}\}_{(u,v) \in V \times V}$ denote the matrix obtained by pair-wise similarity computation,
   where $w_{u,v}$ is a random variable drawn from a probability distribution $f_g$ if $u$ and $v$ belong to the same cluster and drawn from a probability distribution $f_r$ otherwise. The subscripts of $f_r$ and $f_g$ are chosen to respectively signify a ``red edge'' (or absence of a link) and a ``green edge'' (or presence of a link). 
 Note that, this model of similarity matrix is by no means the only possible; however it captures the essential flavor of the problem.

Our main contribution in this paper is to provide a theoretical analysis of query complexities of the two aforementioned 
heuristics from  \cite{wang2013leveraging,vesdapunt2014crowdsourcing}. Our analysis quantifies the effect of the presence of similarity measures 
in these algorithms, establishes the superiority between these algorithms under different criteria, and derives the exact expression of query complexity 
under some fundamental probability models. 

%When comparing to the recently proposed lower bounds to establish the near-optimality or sub-optimality of the above heuristics \cite{mazumdar2016clustering}, interestingly Hellinger divergence between $f_g$ and $f_r$ appears as the right measure of quality of the similarity matrix.

Next, to establish the near-optimality or sub-optimality of the above heuristics, we compare our 
results with an information theoretic lower bound recently proposed by us \cite{mazumdar2016clustering}. As a corollary to the results of \cite{mazumdar2016clustering}, it can be seen that the information theoretic lower bound depends on the Hellinger divergence between $f_g$ and $f_r$. More interestingly, the quality of the similarity matrix can be characterized by the Hellinger divergence between $f_g$ and $f_r$ as well.

% we provide a lower bound on query complexity for the crowdsourced entity resolution problem in terms of the Hellinger divergence between
%the probability measures $f_g$ and $f_r$. We compute the Hellinger distance between some probability models useful in this context and via that establish t.

Finally, we show that the experimental observations of \cite{wang2013leveraging,vesdapunt2014crowdsourcing} agree with our theoretical analysis of their algorithms. Moreover, we conduct a thorough experiment on the bibliographical \texttt{cora}~\cite{cora2004} dataset for ER and several synthetic datasets to validate the theoretical findings further.

\section{System model and techniques}

\subsection{Crowdsourced Entity Resolution \cc}

Consider a set of elements $V\equiv [n]$ which is a disjoint union of $k$ clusters $V_i$, $i =1, \dots, k$, where $k$ and the subsets $V_i \subseteq [n]$  are unknown. 
The crowd (or the oracle) is presented with an element-pair $(u,v) \in V \times V$ for a query, that results in a binary answer denoting the event $u,v$ belonging to the same cluster. 
Note that, this perfect oracle model had been used in the prominent previous works by Wang et al. and Vesdapunt et al.\shortcite{wang2013leveraging,vesdapunt2014crowdsourcing}. 

%While we might model an imperfect oracle as one that gives a wrong answer with some probability independently, such 
%models do not add much to the query complexity of the perfect oracle model. 
We can assume that with probability $0 <p_i <1$, the crowd gives a wrong answer to the $i$th
 query. 
However, with resampling  the $i$th query $\Omega(\log n)$ times, that is by asking the same $i$th query to $\Omega(\log n)$ different users and by taking the majority vote, we can drive the probability $p_i$ to nearly $0$ and return to the model of perfect oracle.
%Indeed, if a crowd answer is faulty with probability $p_{i,j}$, where $i$ is the index of the query and $j$ is index of the number of times the same query being asked, we can take a majority vote among the answers of the crowd and use
 %Chernoff bound to show that we are approaching the perfect oracle model with high probability \cite[Ch.~4]{motwani2010randomized}. This only adds a factor of $\log n$ to the query
%complexity of the 
 %perfect oracle model. 
  Note that we have assumed independence among the resampled queries over the index $j$, which can be justified 
 since we are sampling a growing ($\Omega(\log n)$) number of samples. 
 Furthermore, repetition of the same query to the crowd may not not lead to reduction in the error probability, i.e., a 
 persistent error. Even in this scenario an element can be queried with multiple elements from a same cluster to infer with certainty whether the element belong to the cluster or not.   
 These situations have been covered in detail in our recent work \cite{mazumdar2016clustering}. 
 Henceforth, in this paper, we only consider the perfect oracle model. All our results hold for the faulty oracle model described above with only an $O(\log{n})$ blow-up in the query complexity.

%To make amends for the
%fact that
%iid fault model may not perfectly capture the  correlation among crowd answer, we do not allow resampling the same question. 

 Consider $W$,  an $n \times n$ similarity matrix, with  the $(u,v)$th entry $w_{u,v}$ a nonnegative random variable in $[0,1]$ drawn from a probability density or mass function $f_g$ when ${u,v}$ belong to the same cluster, and  drawn from a probability density or mass function $f_r$ otherwise. $f_g$ and $f_r$ are unknown.

%There is an oracle $\mathcal{O}: V\times V \to \{\pm1\},$ which takes as input a pair of vertices $u,v \in V \times V$, and returns either $+1$ or $-1$. Let $\mathcal{O}(Q)$, $Q \subseteq V \times V$ correspond to oracle answers to all pairwise queries in $Q$. The queries in $Q$ can be done adaptively.

The problem of \cc~is to design a set of queries in $V \times V$, given $V$ and $W$, such that from the answers to the queries, it is possible to 
 recover $V_i$, $i=1,2,...,k$.
  
%\begin{itemize}
%\item {\bf \cc} Here $\mathcal{O}(u,v)=+1$ iff $u$ and $v$ belong to the same cluster and $\mathcal{O}(u,v)=-1$ iff $u$ and $v$ belong to different clusters.
%\begin{enumerate}
%\item {\bf Without Side Information.} Given $V$, find $Q \subseteq V \times V$ such that $|Q|$ is minimum, and from  $\mathcal{O}(Q)$ it is possible to recover $V_i$, $i=1,2,...,k$.
%\item {\bf With Side Information.} Given $V$ and $W$, find $Q \subseteq V \times V$ such that $|Q|$ is minimum, and from  $\mathcal{O}(Q)$ it is possible to recover $V_i$, $i=1,2,...,k$.
%\end{enumerate}
%\end{itemize}
\subsection{The two heuristic algorithms}
{\bf The Edge ordering algorithm \cite{wang2013leveraging}.}
In this algorithm, we arrange the set $V\times V$ in non-increasing order of similarity values $w_{i,j}$s. We then query sequentially according to this order. Whenever possible we apply transitive relation to infer edges. For example, if the queries $(i,j)$ and $(j,l)$  both get positive answers then there must be an edge $(i,l)$, and we do not have to make the query $(i,l)$. We stop when all the edges are either queried, or inferred. %We refer this as {\em edge-ordering} strategy.

\vspace{0.1in}
\noindent{\bf The Node ordering algorithm \cite{vesdapunt2014crowdsourcing}.}
In this algorithm,   the empirical expected size of the cluster containing element $i$, $1 \le i \le n$, is first computed as $\sum_{j} w_{i,j}$.  %first the expected sizes of the clusters are computed assuming the similarity matrix $W$ to give probabilities (for details, see  \cite{vesdapunt2014crowdsourcing}).
Then
 all the elements are ordered non-increasingly according to the empirical expected sizes of the clusters containing them. At any point in the execution, the algorithm maintains at most $k$ clusters. The algorithm selects the next element and issues queries involving that element and elements which are already clustered in non-increasing order of their similarity, and apply transitivity for inference. Therefore, the algorithm issues at most one query involving the current node and an existing cluster. Trivially, this gives an $O(k)$-approximation. 

\remove{
\subsection{Information theoretic notions}
The following information theoretic notions will be useful in our lower bound. The definitions are quite standard and can be found in
many texts such as \cite{gibbs2002choosing}.
\begin{definition}[Total Variation Distance]
For two probability distributions $P$ and $Q$ defined on a sample space $\cX$ and same sigma-algebra $\cF$,
their total variation distance is,
$$
\|P -Q\|_{TV} = \sup \{P(A) -Q(A): A \in \cF\}.
$$
In words, the distance between two distributions is 
their largest difference over any measurable set. For finite $\cX$ total variation distance is half of the $\ell_1$ distance
between pmfs.  
\end{definition}

\begin{definition}[Hellinger Divergence]
The Hellinger divergence between two  distributions $P$ and $Q$ that are absolutely continuous with respect to any other measure
is given by,
$$
\cH(P,Q) = \sqrt{\frac12\int(\sqrt{dP}-\sqrt{dQ})^2}.
$$
If $P$ and $Q$ are continuous distributions with pdfs $f(x)$ and $g(x)$ respectively, then
$$
\cH(f,g) = \sqrt{\frac12\int_{-\infty}^{\infty}(\sqrt{f(x)}-\sqrt{g(x)})^2 dx}.
$$
\end{definition}

Since total variation distance is half of the $\ell_1$-distance, using the inequalities between norms, it can be seen that,
\begin{align}
\label{eq:HTV}
\|P -Q\|_{TV} \le \sqrt{2} \cH(P,Q).
\end{align}

The following property of the Hellinger divergence is going to be useful to us.
Consider a set of  random variables $X_1, \dots, X_m$, and consider the two joint distribution of the 
random variables, $P^m$ and $Q^m$. When the random variables are independent, let $P_i$ and $Q_i$ be the corresponding marginal distribution of the 
random variable $X_i, i =1,\dots, m.$  That is, we have, 
$P^m(x_1, x_2, \dots, x_m ) = \prod_{i =1}^m P_i(x_i)$ and
$Q^m(x_1, x_2, \dots, x_m ) = \prod_{i =1}^m Q_i(x_i).$
Then we must have,
\begin{equation}\label{eq:HJ}
\cH(P^m , Q^m)^2 \le \sum_{i=1}^m \cH(P_i, Q_i)^2.
\end{equation}
}

\section{Information theoretic lower bound}
%Consider the following claim.
%\begin{proposition}
%If there are $k$ clusters in the original dataset, then the  query complexity of  any optimal algorithm for \cc~ is given by
%$(n-k)+{k \choose{2}}.$
%\end{proposition}
%To prove this statement, notice that we need to query at least $(n-k)$  pairs of vertices  to form a spanning tree in each cluster, and ${k \choose{2}}$ pairs to differentiate between every pair of clusters.
Note that, in the absence of similarity matrix $W$, any optimal (possibly randomized) algorithm must make $\Omega(nk)$ queries to 
solve \cc. This is true because an input can always be generated that  makes $\Omega(n)$ vertices to be involved in $\Omega(k)$ queries before
they can be correctly assigned. However, when we are allowed to use the similarity matrix, this  bound can be significantly reduced. 
Indeed, the following lower bound follows as a corollary of the results of
our previous work  \cite{mazumdar2016clustering}. % for a detailed description.

%Our main lower bound result is  summarized below.
\begin{theorem}\label{thm:lb-main}%\label{thm:side}
Given the number of clusters $k$ and $f_g, f_r,$ any randomized algorithm    that does not perform at least $\Omega\Big(\min\Big\{\frac{k^2}{\cH^2(f_g,f_r)}, nk\Big\}\Big)$ queries,  will be unable to return the correct clustering  with high probability,
where $\cH^2(f_g,f_r)\equiv \frac12\int_{-\infty}^{\infty}(\sqrt{f_g(x)}-\sqrt{f_r(x)})^2 dx$ is the squared Hellinger divergence between the probability measures $f_g$ and $f_r$.
  % at least $\frac1{10}$. 
\end{theorem}

%For details, we refer the interested readers to \cite{mazumdar2016clustering}.
The main idea of proving this lower bound already appears in our recent work  \cite{mazumdar2016clustering}, and we give a brief sketch of the proof below for the interested readers. Strikingly, Hellinger divergence between $f_g$ and $f_r$ appears to be the right distinguishing measure even for analyzing the heuristic algorithms.

%Recall that  there are $k$ clusters in the $n$-vertex graph. That is $\cG(V,E)$ is such that, $V = \sqcup_{i=1}^k V_i$ and
%$E = \{(i,j) : i, j \in V_\ell \text{ for some } \ell\}$. In other words, $\cG$ is a union of at most $k$ disjoint cliques. 
%Let $A \equiv [a_{i,j}]$ be the adjacency matrix of $\cG$. 
%Every entry of the side-information matrix $W$ is generated independently as described in the introduction. 

To show the lower bound we consider an input where one of the clusters are fully formed and given to us. The remaining $k-1$ clusters
each has size $a =  \Big\lfloor \frac{1}{8\cH^2(f_g,f_r)}\Big\rfloor.$ 
We prove the result through contradiction. Assume there exists a randomized algorithm ALG that makes a total of 
$o\Big(\frac{k^2}{\cH(f_g,f_r)^2}\Big)$ queries and assigns all the remaining vertices to correct clusters with high probability.
However, that implies that   the average number of queries ALG  makes to 
assign each of the remaining elements to a cluster must be $o(k)$.

Since there are $k$ clusters, this actually guarantees %(see full proof in Appendix \ref{app:lb})
 the  existence of an element that is not queried with  the
correct cluster $V_{i}$ it is from, and that completely relies on the $W$ matrix for the correct assignment. Now the probability 
distribution (which is a product measure) of $W$, $P_W$, can be one of two  different distributions, $P'_W$ and $P''_W$ depending on whether this vertex belong 
to $V_{i}$ or not. Therefore these two distributions must be far apart  in terms of total variation distance for correct assignment.

However, %from Eq.~\eqref{eq:HTV}, 
the total variation distance between $P'_W$ and $P''_W$
$\|P'_W - P''_W\|_{TV} \le \sqrt{2} \cH(P'_W,P''_W)$. But as both $P'_W,P''_W$ are product measures that can differ in at most
$2a$ random variables (recall the clusters are all of size $a$), we must have, using the properties of the Hellinger divergence, %, using  Eq.~\eqref{eq:HJ},
$ \cH(P'_W,P''_W) \le \sqrt{2a \cH(f_g,f_r)^2} \le \frac12$. This means, $\|P'_W - P''_W\|_{TV} \le \frac1{\sqrt{2}}$, i.e.,  
the two distributions are close enough to be confused with a positive probability - which leads to a contradiction.
%A sketch of the proof of this theorem is provided in Section~\ref{sec:lb_sketch}. 
%, where the ingenious idea was to pose the clustering problem as a hypothesis testing problem. We use their technique liberally, and get a stronger bound. We crucially use the Hellinger distance 
%which helps us to derive a universally stronger bound than \cite{mazumdar2016clustering} which relies on  the Kullback-Leibler (KL) divergence instead. 
%In terms of lower bound, the Hellinger distance covers the important cases when divergence is larger than 1. 
%As shown in our examples, this is also reflected in the upper bounds, as in many cases KL distance is unusable (see, {\bf Dist-2} below).
%Also, in many distributions that we consider below, the KL divergence is infinity, resulting in trivial bounds for \cite{mazumdar2016clustering}. Our bound is robust to such pitfalls. 
Note that, in stead of recovery with positive probability, if we want to ensure exact recovery of the clusters (i.e., with probability 1) we must query each element at least once. This leads to the following corollary. 

\begin{corollary}
\label{cor:lb-main}
Any (possibly randomized) algorithm with the knowledge of $f_g, f_r,$ and the number of clusters $k$,  must perform  at least $\Omega\Big(n+\frac{k^2}{\cH^2(f_g,f_r)}\Big)$ queries, $\cH(f_g,f_r)> 0$,  to return the correct clustering exactly. 
\end{corollary}

\section{Main results: Analysis of the heuristics}
We provide expressions for query complexities for both the edge ordering and the node ordering algorithms. It turns out that the following quantity plays a crucial role in the analysis of both:
$$
L_{g,r}(t) \equiv \int_0^1\Big(\int_0^r f_g(y) dy\Big)^tf_r(x)dx.
$$
\begin{theorem}[The Edge ordering]\label{thm:edge}
The query complexity for \cc~ with the edge ordering algorithm is at most,
$$
n+\min_{1\le s \le n} \Big[\binom{k}{2}s^2+ n\sum_{i=1}^{k} \sum_{\ell=s}^{|V_i|} \ell L_{g,r}\Big(\binom{\ell}2\Big)\Big].
$$
\end{theorem}
The proof of this theorem is provided in Section~\ref{sec:edge}.

\begin{theorem}[The Node ordering]\label{thm:node}
The query complexity for \cc~ with the node ordering algorithm is at most,
$$n+\sum_{i=1}^{k}\sum_{s=1}^{|V_i|}  \min\{k, (n-|V_i|)L_{g,r}(s)\}.$$
\end{theorem}
The proof of this theorem is provided in Section~\ref{sec:node}.
\subsection{Illustration: $\epsilon$-biased Uniform Noise Model}

We consider two distributions for $f_r$ and $f_g$ which are only $\epsilon$ far in terms of total variation distance from the uniform distribution. However, if we consider Hellinger distance, then {\bf Dist-1} is closer to uniform distribution than {\bf Dist-2}. These two distributions will be used
as representative distributions to illustrate the potentials of the edge ordering and node ordering algorithms. In both cases, substituting $\epsilon$ with $0$, we get uniform distribution which contains no information regarding the similarities of the entries.

\vspace{0.06in}
\noindent {\bf Dist-1.} Consider the following probability density functions for $f_r$ and $f_g$, where $x \in [0,1]$, and $0 <\epsilon <1/2,$
\[ f_r(x) =
  \begin{cases}
    (1+\epsilon)       &  \text{if } x < \frac{1}{2}\\
    (1-\epsilon)  &  \text{if } x \geq \frac{1}{2}\\
  \end{cases}
\hspace{0.1in} f_g(x) =
  \begin{cases}
    (1-\epsilon)       &  \text{if } x < \frac{1}{2}\\
    (1+\epsilon)  &  \text{if } x \geq \frac{1}{2}.
  \end{cases}
\]
Note that $\int_{0}^{1} f_r(x) \, \diff x=\int_{0}^{1/2} (1+\epsilon)\,  \diff x+\int_{1/2}^{1} (1-\epsilon)\,  \diff x=1$. Similarly, $\int_{0}^{1} f_g(x) \, \diff x=1$, that is they represent valid probability density functions.
We have, $\cH^2(f_g,f_r) = 1-\int_0^1\sqrt{1-\epsilon^2} \diff x = 1-\sqrt{1-\epsilon^2} \approx \epsilon^2/2.$

% $\mu_r\equiv\int_{0}^{1} x f_r(x) \, \diff x=\frac{(1+\epsilon)}{8}+\frac{3(1-\epsilon)}{8}=\frac{2-\epsilon}{4}=\frac{1}{2}-\frac{\epsilon}{4}$ and
%$\mu_g\equiv\int_{0}^{1} x f_g(x) \, \diff x=\frac{(1-\epsilon)}{8}+\frac{3(1+\epsilon)}{8}=\frac{2+\epsilon}{4}=\frac{1}{2}+\frac{\epsilon}{4}$. Here $\mu_r$ and $\mu_g$ respectively represent the mean of the two distributions.

\vspace{0.06in}
\noindent {\bf Dist-2.} Now consider the following probability density functions for $f_r$ and $f_g$ with $0<\epsilon< 1/2$.
\begin{align*}
f_r(x)= \frac{1}{1-\epsilon},  0 \leq x \leq 1-\epsilon, ~~~
f_g(x)=\frac{1}{1-\epsilon},  \epsilon \leq x \leq 1.
\end{align*}
Again, $\int_{0}^{1} f_r(x) \, \diff x=\int_{0}^{1-\epsilon} \frac{1}{(1-\epsilon)}\,  \diff x=1$. Similarly, $\int_{0}^{1} f_g(x) \, \diff x=\int_{\epsilon}^{1} f_g(x) \, \diff x=1$, that is they represent valid probability density functions.
We have, $\cH^2(f_g,f_r) = 1 - \int_{\epsilon}^{1-\epsilon}\frac{1}{1-\epsilon} \diff x = \frac{\epsilon}{1-\epsilon} \approx \epsilon$. 
%$\mu_r\equiv\int_{0}^{1} x f_r(x) \, \diff x=\int_{0}^{1-\epsilon} \frac{x}{(1-\epsilon)} \, \diff x=\frac{1}{2}-\frac{\epsilon}{2}$, and
%$\mu_g\equiv\int_{0}^{1} x f_g(x) \, \diff x=\int_{\epsilon}^{1} \frac{x}{(1-\epsilon)} \, \diff x=\frac{1}{2}+\frac{\epsilon}{2}$.

We have the following results for these two distributions.
\begin{proposition}[Lower bound]\label{thm:lbs}
Any (possibly randomized) algorithm for \cc,~ must make $\Omega(n+\frac{k^2}{\epsilon^2})$ queries for {\bf Dist-1} and $\Omega(n+\frac{k^2}{\epsilon})$ queries for {\bf Dist-2}, to recover the clusters exactly (with probability 1).
\end{proposition}
The proof of this theorem follows from Theorem \ref{thm:lb-main}, Corollary \ref{cor:lb-main}, and by plugging in the Hellinger distances between $f_g, f_r$ in both cases.
%Note that, the lower bound of \cite{mazumdar2016clustering} is  ineffective in the case of {\bf Dist-2}.

The following set of results are  corollaries of Theorem \ref{thm:edge}.
\begin{proposition}[Uniform noise (no similarity information)]
\label{prop:uniform}
Under the uniform noise model where $f_g, f_r \sim Unif[0,1]$, the edge ordering algorithm has query complexity 
$O(nk \log\frac{n}{k})$ for \cc.
%is a 
%$O(\log\frac{\sqrt{n}}{k})$-approximation.
%$\min{\{O(k\ln{\frac{n}{k}}),O(\sqrt{n}\ln{n})\}}$-approximation.
\end{proposition}
\begin{proof}
Since $f_g = f_r$,  the similarity matrix $W$ amounts to no information at all. We know that in this situation, one must make $O(nk)$ queries for the correct solution of \cc.

In this situation, a straight-forward calculation shows that,
$
L_{g,r}(t) = \frac{1}{t+1}.
$
This means, ignoring the first $n$ term, from Theorem \ref{thm:edge}, the edge ordering algorithm makes at most
$
\min_{1\le s \le n} \Big[\binom{k}{2}s^2+ n\sum_{i=1}^{k} \sum_{\ell=s}^{|V_i|} \ell \frac{2}{\ell(\ell-1) +2}\Big] 
\le  \min_{1\le s \le n} \Big[\frac{k^2s^2}{2}+ 2n\sum_{i=1}^{k} \sum_{\ell=s}^{|V_i|} \frac{1}{\ell-1}\Big]
$
number of queries. By bounding  the harmonic series and using the concavity of log, we have the number of queries made by the edge ordering algorithm is at most
$
\min_{1\le s \le n} \Big[\frac{k^2s^2}{2}+ 2n\sum_{i=1}^{k} \ln\frac{|V_i|-1}{s-2}\Big] \\
\le \min_{1\le s \le n} \Big[\frac{k^2s^2}{2}+ 2nk \ln\frac{n-k}{k(s-2)}\Big] = O(nk \log\frac{n}{k}), 
$ 
where we have substituted $s  = \sqrt{n/k}$. % Comparing with the lower bound, this therefore is a $O(\log\frac{\sqrt{n}}{k})$-approximation.
\end{proof}

\begin{proposition}[{\bf Dist-1}]
\label{prop:dist-1}
When $f_g, f_r \sim $ {\bf Dist-1}, the edge ordering algorithm has query complexity $O(nk(1-2\epsilon) \log\frac{n}{k})$ for \cc.
%provides a $\min\{O(k(1-2\epsilon)\ln\frac{\sqrt{n}}{k}),O(\epsilon^2(1-2\epsilon)n\ln\frac{\sqrt{n}}{k}\}$-approximation for \cc.
\end{proposition}
\begin{proof}
The proof is identical to the above.  For small $\epsilon$, we have $L_{g,r}(t) \approx \frac{1-\epsilon}{(1+\epsilon)(t+1)}\approx \frac{1-2\epsilon}{t+1}$ (see, Section~\ref{app:lgrt}).
The algorithm queries at most $O(nk(1-2\epsilon) \log\frac{n}{k})$ edges. 
%Comparison with the lower bound of Theorem \ref{thm:lbs} establishes the proposition.
\end{proof}
\begin{proposition} [{\bf Dist-2}]
\label{prop:dist-2}
When $f_g, f_r \sim $ {\bf Dist-2}, the edge ordering algorithm has query complexity $O\left(n+\frac{k^2\log{n}}{\epsilon} \right)$ for \cc.
\end{proposition}
\begin{proof}
For this case, we have
$L_{g,r}(t) \leq \frac{e^{-\epsilon(t+1)}}{t+1}
$ (see, Section~\ref{app:lgrt}). Choose $s = \sqrt{\frac{4\log{n}}{\epsilon}}+1$. Then using Theorem \ref{thm:edge}, $n\sum_{i=1}^{k} \sum_{\ell=s}^{|V_i|} \ell L_{g,r}\Big(\binom{\ell}2\Big) \leq n\sum_{i=1}^{k} \sum_{\ell=s}^{|V_i|}  \frac{2e^{-\epsilon\binom{\ell}2}}{\ell-1}< 1$. Therefore, the number of queries is $O\left(n+\frac{k^2\log{n}}{\epsilon} \right)$, matching the lower bound within a $\log{n}$ factor.
\end{proof}

For the Node-ordering algorithm, we have the following result as a corollary of Theorem~\ref{thm:node}.

\begin{proposition} [{\bf Node-Ordering}]
\label{prop:dist-3}
When $f_g, f_r \sim $ {\bf Dist-1}, the node ordering algorithm has query complexity $O(nk(1-\epsilon^2))$ for \cc.
When $f_g, f_r \sim $ {\bf Dist-2},  node ordering  has query complexity $O\left(n+\frac{k^2\log{n}}{\epsilon}\right)$ for \cc.
\end{proposition}
\begin{proof}
For {\bf Dist-1}, $L_{g,r}(s)\approx \frac{1-2\epsilon}{s+1}$. Therefore, when $s\geq\frac{n}{k}(1-\epsilon)$, $\min{\{k, (n-|V_i|)L_{g,r}(s)\}}\leq(1-\epsilon)k$. Thus, the total number of queries is $O(nk(1-\epsilon)+\epsilon nk(1-\epsilon))=O(nk(1-\epsilon^2))$. 
For {\bf Dist-2}, $L_{g,r}(s)=\frac{\exp(-\epsilon s)}{s+1}$. Therefore, when $s\geq \frac{2\log{n}}{\epsilon}$, $\min{\{k, (n-|V_i|)L_{g,r}(s)\}} \leq \frac{1}{n(s+1)}$. Thus the total number of expected queries is $O(n+\sum_{i=1}^{k}\frac{k\log{n}}{\epsilon}+\frac{|V_i|\log{|V_i|}}{n})=O(n+\frac{k^2\log{n}}{\epsilon})$. 
\end{proof}

Note that, there is no difference in the upper bounds given between the Edge and Node ordering algorithms for {\bf Dist-2}.  But  Edge-ordering  uses order $\log(n/k)$ factor more queries than the optimal ($O(nk)$) for {\bf Dist-1}.
{\bf Dist-1} is closer to uniform distribution by the Hellinger measure than {\bf Dist-2},  which shows that Hellinger distance is the right choice for distance here.  Assuming $k=o(n)$, we get a drastic reduction in query complexity by moving from {\bf Dist-1} to {\bf Dist-2}.

\begin{figure*}[ht]
	\centering
	\vspace{-0.5in}
\subfloat[\texttt{similarity value distribution}]{\includegraphics[width=0.45\textwidth]{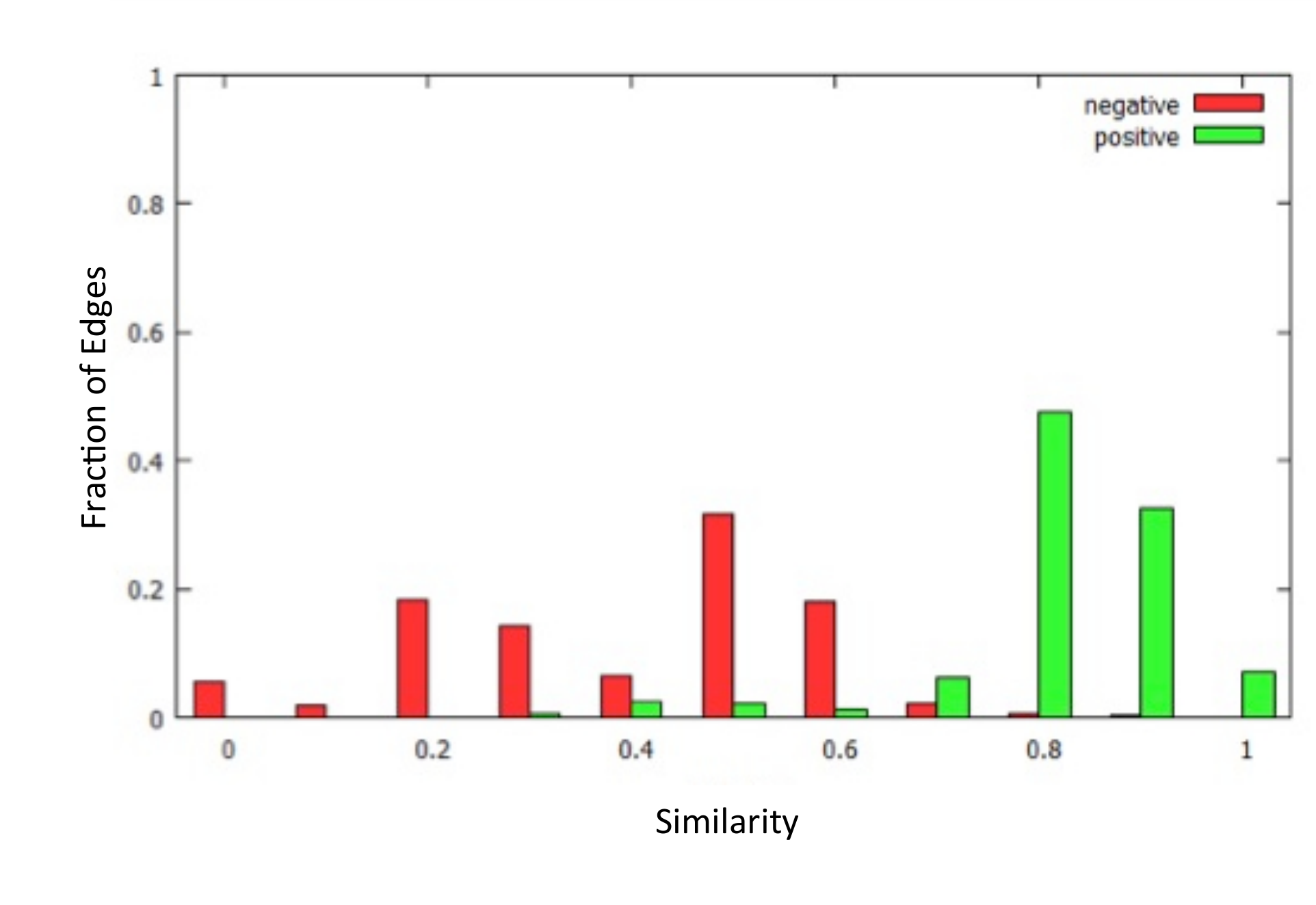}}
%\hspace{0.2cm}
\subfloat[\texttt{ \#queries vs recall}]{\includegraphics[width=0.4\textwidth]{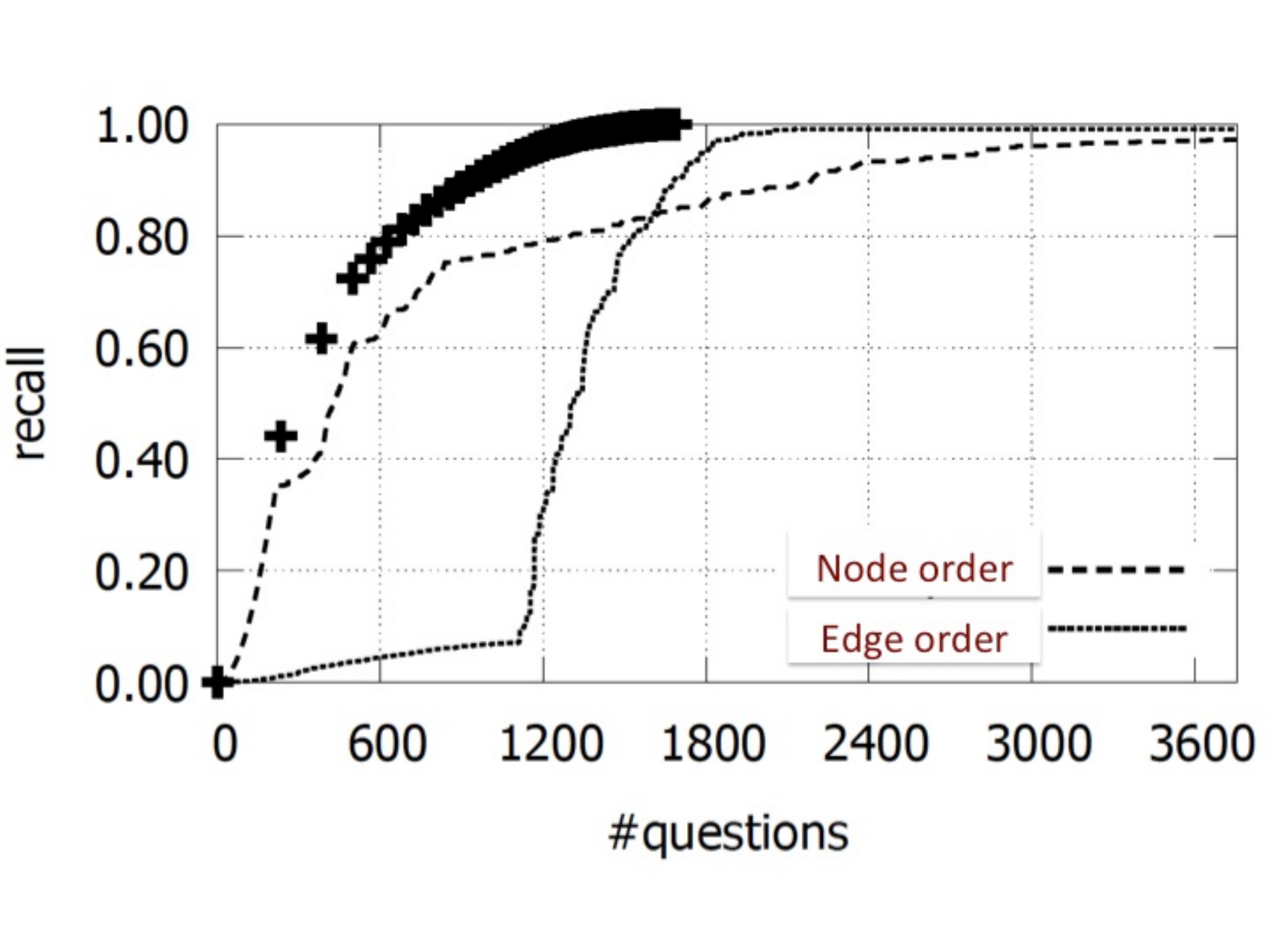}}
\caption{\texttt{cora}
	\label{fig:graph}}
\end{figure*}

\section{Analysis of the Edge ordering algorithm: proof of Theorem \ref{thm:edge}}\label{sec:edge}
Let $R$ be a random variable with distribution $f_r$ and $G_1, \dots, G_t$ be identical random variables with distribution $f_g$. Let $R, G_1, G_2, \dots, G_t$ be all independent. 
Note that,
\begin{align}
&\Pr(R \ge \max\{G_1, \dots, G_t\})\nonumber \\
&= \int_0^1\Big(\int_0^r f_g(y) dy\Big)^tf_r(x)dx = L_{g,r}(t) .
\end{align}

In the interest of clarity, let us call a pair $(u,v) \in V\times V$ a {\em green} edge iff $u, v \in V_i$ for some $i =1, \dots, k$, and otherwise call the pair a {\em red edge}.

In the current graph, let there exist $\ell$ nodes, called $U\subset V$, which all belong to the same cluster but no edge from the induced graph on these $\ell$ vertices have been queried yet. Then there are $\binom{\ell}2$ green edges within $U$, yet to be queried. On the other hand, there are at most $n\ell$ red edges with one end point incident on the vertices in $U$. We now count the number of red edges incident on $U$ that the algorithm will query before querying a green edge within  $U$. We can account for all the red edges queried by the algorithm by considering each cluster at a time, and summing over the queried red edges incident on it. In fact, by doing this, we double count every red edge. Since the probability of querying a red edge incident on $U$ before querying any of the $\binom{\ell}2$ green edges incident on $U$ is $L_{g,r}(\binom{\ell}2$, the 
expected number of  queried red edges incident on $U$ before querying a green edge in $U$  is at most $ n\ell L_{g,r}(\binom{\ell}2).$ 

Let $s$ be a positive integer. Consider a cluster $V_i: |V_i| \ge s$. Suppose at some point of time, there are $\ell$ components of $V_i$ remaining to be connected. Then, again there are at least $\binom{\ell}2$ green edges, querying any of which will decrease the number of components by $1$. Thus, 
 the expected number of red edges that are queried incident on nodes in $V_i$ before there remain at most $s$ components of $V_i$ is at most $n\sum_{\ell=s}^{|V_i|} \ell L_{g,r}(\binom{\ell}2)$.
%We have 
%\begin{align*}
%\sum_{R=\sqrt{\frac{n}{k}}}^{|C_i| } \frac{n}{R} &\leq  n \int_{R=\sqrt{\frac{n}{k}}}^{|C_i|}\frac{1}{R}dR \\
%&= n(\ln{|C_i|}-\ln{\big(\sqrt{\frac{n}{k}}\big)})=n\ln{\frac{|C_i|}{\sqrt{\frac{n}{k}}}}
%\end{align*}
Therefore, the expected number of red edges that are queried until only $s$ components are left for every cluster is 
%\begin{align*}
$n\sum_{i=1}^{k} \sum_{\ell=s}^{|V_i|} \ell L_{g,r}\Big(\binom{\ell}2\Big).$
%\end{align*}

Now the number of red edges across the clusters  having size at most $s$ is at most
$\binom{k}{2}s^2.$
Therefore, even if we query all those edges, we get 
the total number of queried red edges to be at most $\binom{k}{2}s^2+ n\sum_{i=1}^{k} \sum_{\ell=s}^{|V_i|} \ell L_{g,r}\Big(\binom{\ell}2\Big).$

The algorithm queries a total of $n-k$ green edges, exactly spanning every cluster. Thus the total number of queries is at most $n+\binom{k}{2}s^2+ n\sum_{i=1}^{k} \sum_{\ell=s}^{|V_i|} \ell L_{g,r}\Big(\binom{\ell}2\Big).$

\remove{

\subsection{Analysis of Wang et al. \cite{} Edge-ordering Strategy}

We are now ready to analyze the edge-ordering strategy.

%Suppose $G=(V,E)$ consists of $k$ clusters $\calC_1,\calC_2,\calC_3,...,\calC_k$. Let their sizes be respectively $c_1,c_2,..,c_k$. Then the total number of green edges in $G$ is
%\begin{align}
%\label{eq:green}
%\sum_{i=1}^{k} {c_i \choose 2} =\frac{1}{2}\sum_{i=1}^{k}c_i^2-c_i&=\left(\frac{1}{2}\sum_{i=1}^{k}c_i^2 \right)-\frac{n}{2} \nonumber\\
%&\geq \frac{1}{2k} \left(\sum_{i=1}^{k} c_i\right)^2-\frac{n}{2} ~~\text{by Cauchy Schwartz}\nonumber\\
%&=\frac{n^2}{2k}-\frac{n}{2}
%\end{align}
%
%We know by transitive relation if two vertices $a$ and $b$ are connected by a path of only green edges, and if  $query(a,u)=RED$ for some vertex $u$ not in that path, then that implies $query(b,u)$ is also RED and the algorithm never queries $(b,u)$. Similarly, if $query(a,u)=GREEN$, then $query(b,u)$ must also be GREEN and again the algorithm never queries it. Hence, we can merge $a$ and $b$, and all other vertices in the same component formed by the green edges and treat them as a single super-vertex subsequently. The above relation (Eq. \ref{eq:green}) holds in the new contracted graph. 

%\begin{proposition}
%\label{prop:uniform}
%Under the uniform noise model where $f_g, f_r \sim Unif[0,1]$, the edge-ordering strategy is a $\min{\{O(k(\ln{n}-\ln{k})),O(\sqrt{n}\ln{n})\}}$-approximation.
%\end{proposition}
\begin{proof}
%When $f_g, f_r \sim Unif[0,1]$, we can set $\epsilon=0$ for the computation of $Prob(R \geq \max(G_1,G_2,..,G_t))$ to obtain $Prob(R \geq \max(G_1,G_2,..,G_t)) =\frac{1}{t+1}$. 

In the current graph, let there exist $R$ nodes which all belong to the same cluster but no edge from the induced graph on these $R$ vertices have been queried yet. Then there are $O(R^2)$ green edges incident on these vertices, yet to be queried. On the other hand, there are at most $nR$ red edges with one end point incident on the vertices in $R$. We now count the number of red edges incident on $R$ that the algorithm will query before querying a green edge incident on $R$. We can account for all the red edges queried by the algorithm by considering each cluster at a time, and summing over the queried red edges incident on it. In fact, by doing this, we double count every red edge.
$$\expect[\# \text{ queried red edges incident on $R$ before querying a green edge incident on $R$} ] = O\left( \frac{nR}{R^2}\right)=O\left( \frac{n}{R}\right)$$

Now consider the $i$th cluster. Suppose $|C_i| > \sqrt{\frac{n}{k}}$. The expected number of red edges that are queried incident on nodes in $C_i$ before there remain at most $\sqrt{\frac{n}{k}}$ components of $C_i$ is $O\left(\sum_{R=|C_i|}^{\sqrt{\frac{n}{k}}} \frac{n}{R}\right)$.

We have 
\begin{align*}
\sum_{R=\sqrt{\frac{n}{k}}}^{|C_i| } \frac{n}{R} &\leq  n \int_{R=\sqrt{\frac{n}{k}}}^{|C_i|}\frac{1}{R}dR \\
&= n(\ln{|C_i|}-\ln{\big(\sqrt{\frac{n}{k}}\big)})=n\ln{\frac{|C_i|}{\sqrt{\frac{n}{k}}}}
\end{align*}

Therefore, the expected number of red edges that are queried until only $\sqrt{\frac{n}{k}}$ components are left for every cluster is 
\begin{align*}
n\sum_{i=}^{k} \ln{\frac{|C_i|}{\sqrt{\frac{n}{k}}}}&=nk \ln{\left(\prod_{i=1}^{k} \frac{|C_i|}{\sqrt{\frac{n}{k}}}\right)^{\frac{1}{k}}}\\
&\leq nk \ln{\sum_{i=1}^{k} \frac{|C_i|}{k\sqrt{\frac{n}{k}}}}=nk\ln{\sqrt{\frac{n}{k}}}
\end{align*}

Now the number of red edges across the clusters  having size at most $\sqrt{\frac{n}{k}}$ is 
$\leq k^2\frac{n}{k}=nk$

Therefore, even if we query all those edges, we get 
the total number of queries to be $O\left(nk(\ln{n}-\ln{k})\right)$

The algorithm queries a total of $n-k$ green edges, exactly spanning every cluster. Thus the total number of queried edges is $O(nk(\ln{n}-\ln{k}))$, giving an $O(k(\ln{n}-\ln{k}))$-approximation.

On the other hand, if $k \leq \sqrt{n}$, then this gives a $O(\sqrt{n}\ln{n})$-approximation. If $k > \sqrt{n}$, then any optimal algorithm must query $O(k^2)$ edges, and we have $nk \leq \sqrt{n} k^2$. Thus under uniform noise model, the edge-ordering strategy gives an $O(\sqrt{n}\ln{n})$-approximation.
\end{proof}
}

\section{Analysis of the Node ordering algorithm: proof of Theorem \ref{thm:node}}\label{sec:node}
The computed expected cluster size for each node can be a highly biased estimator, and may not provide any useful information. For example, the expected cluster size of a node in $V_i$ is $\frac{\epsilon}{c} |V_i|+(\frac{1}{2}-\frac{\epsilon}{2c})n$ where $c=2$ for {\bf Dist 1} and $c=1$ for {\bf Dist 2}. Therefore, the node ordering considered by \cite{vesdapunt2014crowdsourcing} can be arbitrary. Hence, for the purpose of our analysis, we ignore this ordering based on the expected size.

Consider the state of the algorithm where it needs to insert a node $v$ which truly belongs to cluster $V_i$. Suppose the current size of $V_i$ is $s$, that is $V_i$ already contains $s$ nodes when $v$ is considered. Consider another cluster $V_j$, $j \neq i$, and let its current size be $s'$. Let $C_i$ and $C_j$ denote the current subclusters of $V_i$ and $V_j$ that have been formed.Then, $P(w_{v,u} \geq \max_{x \in V_i} w_{v,x})$ where $u \in C_j$ is at most $L_{g,r}(s)$. Hence, $P(\exists u \in C_j, w_{v,u} \geq \max_{x \in V_i} w_{v,x})\leq \min{\{1,s'L_{g,r}(s)\}}$. Thus the expected number of queried red edges before $v$ is correctly inserted in $V_i$ is at most $\min{\{k, L_{g,r}(s)\sum_{j \in [1,k], j \neq i}|V_j|\}}\leq \min\{k, (n-|V_i|)L_{g,r}(s)\}$. Hence the expected total number of queried red edges to grow the $i$th cluster is at most $\sum_{s=1}^{|V_i|}  \min\{k, (n-|V_i|)L_{g,r}(s)\}$, and thus the expected total number of queries, including green and red edges is bounded by $n+\sum_{i=1}^{k}\sum_{s=1}^{|V_i|}  \min\{k, (n-|V_i|)L_{g,r}(s)\}$.

\section{Experimental Observations} 

A detailed comparison of the node ordering and edge ordering methods on multiple real datasets has been shown in \cite[Figures 12,14]{vesdapunt2014crowdsourcing}. %, which motivated this paper. 
The number of queries issued by the two methods are very close on complete resolution.To validate further, we did the following experiments. 

\noindent{\bf Datasets.}
(i) We created multiple synthetic datasets each containing $1200$ nodes and $14$ clusters with the following size distribution: two clusters of size $200$, four clusters of size $100$, eight clusters of size $50$, two clusters each of size $30$ and $20$ and the rest of the clusters of size $10$. The datasets differed in the way similarity values are generated by varying $\epsilon$ and sampling the values either from {\bf Dist-1} or {\bf Dist-2}. The similarity values are further discretized to take values from the set $\{0,0.1,0.2,...,0.9,1\}$.

(ii) We used the widely used \texttt{cora}~\cite{cora2004} dataset for ER.  \texttt{cora} is a bibliography dataset, where
each record contains title, author, venue, date, and pages attributes. There are $1878$ nodes in total with $191$ clusters, among which $124$ are non-singletons. The largest cluster size is $236$, and the total number of pairs is $17,64,381$. We used the similarity function as in ~\protect\cite{whang2013question,wang2013leveraging,vesdapunt2014crowdsourcing,fss:16}. 
%Figure \ref{fig:graph}(a) shows the similarity value distribution for \texttt{cora} which is close to {\bf Dist-2}.

\noindent{\bf Observation.}
% Requires the booktabs if the memoir class is not being used
\begin{table}[htbp]
   \centering
   \begin{tabular}{|l|l|l|l|} % Column formatting, @{} suppresses leading/trailing space
   \hline
     Node-Ordering    & Edge-Ordering & Distribution & $\epsilon$\\
      \hline
      4475     & 4460 & {\bf Dist-1} & $\epsilon=\frac{1}{2}$ \\
      5207    &  6003   &  {\bf Dist-1} & $\epsilon=\frac{1}{3}$ \\
      5883       & 7145  & {\bf Dist-1} & $\epsilon=\frac{1}{4}$ \\
      6121       & 7231  & {\bf Dist-1} & $\epsilon=\frac{1}{5}$ \\
      6879 & 8545   &  {\bf Dist-1} & $\epsilon=\frac{1}{10}$ \\
      7398 & 9296 &{\bf Dist-1} & $\epsilon=\frac{1}{20}$\\
      \hline
      1506 & 1277 & {\bf Dist-2} & $\epsilon=\frac{1}{5}$ \\
1986 & 1296 & {\bf Dist-2} & $\epsilon=\frac{1}{10}$\\
2760 & 1626 & {\bf Dist-2} & $\epsilon=\frac{1}{20}$\\
\hline
   \end{tabular}
   \caption{Number of Queries for {\bf Dist-1} and {\bf Dist-2}}
   \label{tab:queries}
\end{table}
The number of queries for the node-ordering and edge-ordering algorithms are reported in Table \ref{tab:queries} for the synthetic datasets. Clearly, the number of queries asked for {\bf Dist-2} is significantly less than that for {\bf Dist-1} at the same value of $\epsilon$. This confirms with our theoretical findings. Interestingly, we observe that the number of queries asked by the edge-ordering algorithm is consistently higher than the node-ordering algorithm under {\bf Dist-1}. This is also expected from Propositions~\ref{prop:dist-1} and ~\ref{prop:dist-3} due to a gap of $\log{\frac{n}{k}}$ in the number of queries of the two algorithms. In a similar vein, we see the edge-ordering algorithm is more effective than the node-ordering for {\bf Dist-2}, possibly because of hidden constants in the asymptotic analysis. 

Figure \ref{fig:graph}(a) shows the similarity value distribution for \texttt{cora} which is closer to {\bf Dist-2} than {\bf Dist-1}. 
Figure  \ref{fig:graph}(b) shows the recall vs number of queries issued by the two methods. 
The line marked with  `+' sign is the curve for the ideal algorithm that will ask only the required ``green'' edges first to grow all the clusters and then ask just one ``red'' edge across every pair of clusters. 
Upon completion, the number of queries issued by the edge ordering and node ordering methods are respectively 21,099 and 23,243 which are very close to optimal. Interestingly, this confirms with our observation on the However, they achieve above $0.996$ recall in less than $4,000$ queries. This can also be explained by our analysis. The remaining large number of queries are mainly spent on growing small clusters, e.g. when cluster sizes are $o(\log{n})$--they do not give much benefit on recall, but consume many queries.

\section{Appendix:  $L_{g,r}(t)$ for {\bf Dist-1}, {\bf Dist-2}} \label{app:lgrt}
%\section{Querying strategy: upper bound on query complexity}
%\section{Analysis}
%Let $G_1,G_2,...,G_t\sim G$ for some $t \geq 1$. For both the distributions, we obtain an upper bound on $$\Pr(R \geq \max_{t}(G_1,G_2,...,G_t))$$ which will be useful for our analysis.
\begin{proposition}
For $f_g,f_r \sim$ {\bf Dist-1} and small $\epsilon$, we have $L_{g,r}(t) \approx \frac{ (1-\epsilon)}{(1+\epsilon)(t+1)}$.  % if $f_r$ and $f_g$ are drawn according to Dist-1.
\end{proposition}
\begin{proof} We have,
\begingroup
\allowdisplaybreaks
\begin{align*}
&L_{g,r}(t)
=\int_{r=0}^{1}\left( \int_{x=0}^{r}f_{G}(x)\,\diff x \right)^t f_{R}(r)\, \diff r\\
&=\int_{r=0}^{1/2}\left( \int_{x=0}^{r}f_{G}(x)\,\diff x \right)^t (1+\epsilon)\, \diff r \\
& \hspace{0.2in}+ \int_{r=1/2}^{1}\left( \int_{x=0}^{r}f_{G}(x)\,\diff x \right)^t (1-\epsilon)\, \diff r\\
&=\int_{r=0}^{1/2}\left( \int_{x=0}^{r}(1-\epsilon)\,\diff x \right)^t (1+\epsilon)\, \diff r 
+ \int_{r=1/2}^{1} \\ & \left( \int_{x=0}^{1/2}(1-\epsilon) \,\diff x + \int_{x=1/2}^{r}(1+\epsilon) \,\diff x \right)^t (1-\epsilon)\, \diff r\\
%&=(1+\epsilon)\int_{r=0}^{1/2} r^t(1-\epsilon)^t \, \diff r  \\
%&+  (1-\epsilon)\int_{r=1/2}^{1}\left( (1-\epsilon)\frac{1}{2} + (1+\epsilon)\left(r-\frac{1}{2}\right) \right)^t\, \diff r\\
%&=(1+\epsilon)\int_{r=0}^{1/2} r^t(1-\epsilon)^t \, \diff r \\
% & \hspace{0.2in}+ (1-\epsilon)\int_{r=1/2}^{1}\left( r(1+\epsilon)-\epsilon \right)^t\, \diff r\\
%&=(1+\epsilon)(1-\epsilon)^t\frac{r^{t+1}}{(t+1)} \Big|_{0}^{1/2} \\
% &\hspace{0.2in}+ (1-\epsilon)\int_{r=1/2}^{1}\left( r(1+\epsilon)-\epsilon \right)^t\, \diff r\\
&=\frac{(1+\epsilon)(1-\epsilon)^t}{2^{t+1}(t+1)}+ (1-\epsilon)\int_{r=1/2}^{1}\left( r(1+\epsilon)-\epsilon \right)^t\, \diff r
 \end{align*}
\endgroup
 Set $z=r(1+\epsilon)-\epsilon$, then $\diff z=(1+\epsilon) \diff r$. We have
 \begin{align*}
 (1-\epsilon)&\int_{r=1/2}^{1}\left( r(1+\epsilon)-\epsilon \right)^t\, \diff r 
 =\frac{ (1-\epsilon)}{(1+\epsilon)}
 \int_{z=\frac{(1-\epsilon)}{2}}^{1} z^t \, \diff z \\
 &=\frac{ (1-\epsilon)}{(1+\epsilon)(t+1)} \left(1-\frac{(1-\epsilon)^{t+1}}{2^{t+1}}\right).
 \end{align*}
 Therefore,
 \begingroup
\allowdisplaybreaks
 \begin{align*}
&L_{g,r}(t)\\  %\Pr(R \geq \max_{t}(G_1,G_2,...,G_t)) \\
&=\frac{(1+\epsilon)(1-\epsilon)^t}{2^{t+1}(t+1)}+\frac{ (1-\epsilon)}{(1+\epsilon)(t+1)} \left(1-\frac{(1-\epsilon)^{t+1}}{2^{t+1}}\right) \\
%&=\frac{ (1-\epsilon)}{(1+\epsilon)(t+1)} +\frac{(1-\epsilon)^t}{2^{t+1}(t+1)}\left( (1+\epsilon)-\frac{(1-\epsilon)^2}{(1+\epsilon)}\right)\\
%&=\frac{ (1-\epsilon)}{(1+\epsilon)(t+1)} +\frac{(1-\epsilon)^t}{2^{t+1}(t+1)}\frac{4\epsilon}{(1+\epsilon)}\\
&=\frac{ (1-\epsilon)}{(1+\epsilon)(t+1)}\left(1+\epsilon\left(\frac{1-\epsilon}{2}\right)^{t-1}\right).
\end{align*}
\endgroup
\end{proof}
\begin{proposition}
For $f_g,f_r \sim$ {\bf Dist-2} we have $L_{g,r}(t) \leq \frac{\e^{-\epsilon(t+1)}}{t+1}$. % if $f_r$ and $f_g$ are drawn according to Dist-2.
\end{proposition}
\begin{proof}
We have,
\begingroup
\allowdisplaybreaks
 \begin{align*}
&L_{g,r}(t)
=\int_{r=\epsilon}^{1-\epsilon} \left( \int_{x=\epsilon}^{r} f_{G}(x)\, \diff x \right)^t \frac{1}{1-\epsilon} \, \diff r\\
&=\int_{r=\epsilon}^{1-\epsilon} \left( \int_{x=\epsilon}^{r} \frac{1}{1-\epsilon} \, \diff x \right)^t \frac{1}{1-\epsilon} \, \diff r \\
&= \frac{1}{(1-\epsilon)^{t+1}} \int_{r=\epsilon}^{1-\epsilon} (r-\epsilon)^t \,\diff r 
= \frac{1}{t+1}\left(\frac{1-2\epsilon}{1-\epsilon}\right)^{t+1}\\
&=\frac{1}{t+1}\left(1-\frac{\epsilon}{1-\epsilon}\right)^{t+1} \leq \frac{(1-\epsilon)^{t+1}}{t+1}\le \frac{\e^{-\epsilon(t+1)}}{t+1}.
\end{align*}
\endgroup
\end{proof}

%\vspace{0.1in}

\noindent{\em Acknowledgements:} This research is supported in part by NSF CCF Awards 1464310, 1642658, 1642550 and a Google Research Award. The authors would like to thank Sainyam Galhotra for his many help with the simulation results.

\bibliographystyle{aaai}

%{\small
\bibliography{bibfile}

\end{document}